\documentclass[reqno,11pt,tikz,border=3pt]{amsart}
\newtheorem{definition}{Definition}
\usepackage{graphicx}
\usepackage{xcolor}
\usepackage{tikz}
\usepackage{tikz-cd}
\usepackage{caption}
\usepackage{booktabs}
\usepackage{adjustbox}
\usepackage{longtable}
\usepackage{array}
\usepackage{amssymb}

\usepackage{amsmath}
\newcommand{\be}{\begin{equation}}
\newcommand{\ee}{\end{equation}}
\newcommand{\dis}{\displaystyle}

\DeclareMathSymbol{\Lambda}{\mathord}{operators}{"03}

\usepackage{enumerate}% http://ctan.org/pkg/enumerate
\usepackage{enumitem}

\usepackage{appendix}

\newtheorem{thm}{Theorem}[section]
\newtheorem{prop}[thm]{Proposition}
\newtheorem{theorem}[thm]{Theorem}

\newtheorem{remark}[thm]{\it Remark}

\usepackage{tikz-3dplot}
\usepackage{xifthen}
\usepackage{caption}

\tdplotsetmaincoords{60}{125}
\tdplotsetrotatedcoords{0}{20}{0} %<- rotate around (z,y,z)

\usepackage[
labelfont=sf,
hypcap=false,
format=hang
]{caption}

\usepackage{mathtools}

\usetikzlibrary{calc,backgrounds,plotmarks}

\pgfdeclarelayer{background}
\pgfdeclarelayer{foreground}
\pgfsetlayers{background,main,foreground}

\begin{document}

\title{Idempotent compatible maps and discrete integrable systems on the triangular lattice}

\author[P. Kassotakis]{Pavlos Kassotakis}\footnote{Corresponding author: P. Kassotakis}
\address{Pavlos Kassotakis, Department of Mathematical Methods in Physics, Faculty of Physics,
University of Warsaw, Pasteura 5, 02-093, Warsaw, Poland}
 \email{Pavlos.Kassotakis@fuw.edu.pl, pavlos1978@gmail.com}

\author[M. Nieszporski]{Maciej Nieszporski}
\address{Maciej Nieszporski, Department of Mathematical Methods in Physics, Faculty of Physics,
University of Warsaw, Pasteura 5, 02-093, Warsaw, Poland}
 \email{Maciej.Nieszporski@fuw.edu.pl}

\date{\today}

\date{\today}

\begin{abstract}
%We classify rational non-invertible compatible maps of a specific form.
We present three equivalence classes of rational non-invertible multidimensional compatible maps.
 These maps turns out to be idempotent and by construction they admit birational partial inverses (companion maps) which are Yang-Baxter maps. 
  The maps in question can be reinterpreted as systems of difference equations defined on the edges of the $\mathbb{Z}^2$ graph.   
 Finally, we associate these compatible systems of difference equations with  integrable difference equations defined on the triangular lattice $Q(A2)$.
\end{abstract}

\maketitle

\section{Introduction}\label{sec1}

%Birational maps on algebraic varieties that admit a compatibility property, play a prominent role in the theory of discrete integrable systems %\cite{Noumi1998,ABS:YB,pap2-2006,Franks-book}. A prototypical example of such a birational map is the map
%$R:(x,y)\mapsto\left(y-\frac{p-q}{x+y},x+\frac{p-q}{x+y}\right),$ $x,y,p,q\in \mathbb{CP}^1,$ see \cite{adler-1993}, whereas its compatibility property is %manifested by the Yang-Baxter equation.
%\begin{align*}
%R:(x,y)\mapsto &\left(y-\frac{p-q}{x+y},x+\frac{p-q}{x+y}\right),& x,y,p,q\in &\mathbb{CP}^1,
%\end{align*}
%The map $R$ is related to the discrete analogue of the potential KdV equation \cite{Nijhoff:1995,WaEs}  which reads
%\begin{align*}
%  (\phi_{m+1,n+1}-\phi_{m,n})(\phi_{m+1,n}-\phi_{m,n+1})=&p_m-q_n, & m,n\in \mathbb{Z}
%\end{align*}
%where $\phi:\mathbb{Z}^2\rightarrow \mathbb{CP}^1,$ $p,q:\mathbb{Z}\rightarrow \mathbb{CP}^1.$

The aim of this article is to initiate the study on a certain  class of non-birational maps which are also related to integrable partial difference equations. These maps turn out to be idempotent, hence non-birational, and they are multidimensional compatible, see Section \ref{sec2}. 
 An example of an idempotent multidimensional compatible map, which is introduced in this article and it is referred to as $Q_I$, reads
 \begin{align*}
 Q_I:(x,y)\mapsto & \left((p-q)\frac{x}{x-y},(p-q)\frac{y}{x-y}\right),
 \end{align*}
where $x,y,p,q\in \mathbb{CP}^1,$ and as we shall see it is related to partial difference equation
\begin{align}
(\phi_{m+1,n+1}+p_m)(\phi_{m,n+1}+q_n)=&(\phi_{m+1,n+1}+q_n)(\phi_{m+1,n}+p_n),& m,n\in&\mathbb{Z}
\end{align}
that serves as the discrete analogue of the Burgers equation \cite{Levi:1983}, c.f. \cite{Zhang:2021,Zhang:2022}.

The companion map, known also as  partial inverse (see Section \ref{sec2}) of $Q_I$, is a birational map that serves as a set theoretical solution of the Yang-Baxter equation and explicitly reads
\begin{align*}
R_{I}:(x,y) \mapsto &  \left(\frac{x y}{x+q-p},x+q-p\right).
\end{align*}
The role of idempotency in the context of set theoretical solutions of the Yang-Baxter  equation, has been  studied in \cite{Atkinson2013,Colazzo2023}, while idempotent set theoretical solutions of the pentagon  equation have been recently obtained in \cite{Mazzotta2024}. In this article we investigate  the role of idempotency in the theory of discrete integrable systems. %In detail, we present  three equivalence classes of rational idempotent multidimensional compatible maps. % which admit birational partial inverses (companion maps) which are Yang-Baxter maps. The

 The article is organised as follows. 
 In Section \ref{sec2}, we classify up to Möbius transformations rational non-invertible  compatible maps of a specific form. These maps fall into three equivalence classes of rational idempotent multidimensional compatible maps. Each class admits birational partial inverses which satisfy the Yang-Baxter equation. In Section \ref{sec3}, we demonstrate that these maps correspond to linearizable partial difference equations defined on the edges of the
$\mathbb{Z}^2$  graph. We then associate these edge-variable equations with vertex-variable equations defined on a triangular stencil and we obtain  difference systems in vertex variables  defined on  “black” and ``white" triangles of
triangular $Q(A2)$ lattices. 

\section{Idempotent compatible maps} \label{sec2}

\subsection{ Yang-Baxter  and $3D-$compatible quadrirational maps  } \label{section2}
%%%%%%%%%%%%
 %We recall the following definitions. 
 
 Let $\mathcal{X}$ be any set.
 \begin{definition}[$3D-$compatible maps \cite{ABS:YB}]\label{Def1}
 Let $Q: \mathcal{X} \times \mathcal{X}\ni({\bf x},{\bf y})\mapsto ({\bf u}, {\bf v})=(f({\bf x},{\bf y}),g({\bf x},{\bf y})) \in \mathcal{X} \times \mathcal{X},$ be a map and  $Q_{ij}$  $i\neq j\in\{1,2,3\},$ be the maps that act as $Q$ on the $i-$th and $j-$th factor of $\mathcal{X} \times \mathcal{X}\times \mathcal{X}$ and as identity to the remaining factor.  In detail we have
\begin{align*}
Q_{12}:({\bf x},{\bf y},{\bf z})\mapsto( {\bf x}_2, {\bf y}_1,{\bf z})=(f({\bf x},{\bf y}),g({\bf x},{\bf y}),{\bf z}),\\
Q_{13}:({\bf x},{\bf y},{\bf z})\mapsto( {\bf x}_3,{\bf y},{\bf z}_1)=(f({\bf x},{\bf z}),{\bf y},g({\bf x},{\bf z})),\\
Q_{23}:({\bf x},{\bf y},{\bf z})\mapsto({\bf x}, {\bf y}_3, {\bf z}_2)=({\bf x},f({\bf y},{\bf z}),g({\bf y},{\bf z})).
\end{align*}

The map $Q: \mathcal{X} \times \mathcal{X}\rightarrow \mathcal{X} \times \mathcal{X}$ will be called {\em 3D-compatible} or {\em 3D-consistent map}  if it holds
${{\bf x}_{23}}={ {\bf x}_{32}},$  ${ {\bf y}_{13}}={{\bf y}_{31}},$ ${{\bf z}_{12}}={{\bf z}_{21}},$ that is
\begin{align}\label{3d:comp:def1}
f( {\bf x}_3,{\bf y}_3)=f({\bf x}_2,{\bf z}_2),&&g( {\bf x}_3, {\bf y}_3)=f( {\bf y}_1, {\bf z}_1),&&g( {\bf x}_2, {\bf z}_2)=g( {\bf y}_1,{\bf z}_1).
\end{align}
\end{definition}

\begin{definition}[Yang-Baxter maps \cite{Sklyanin:1988,Drinfeld:1992}]
A map $R: \mathcal{X} \times \mathcal{X}\ni({\bf x},{\bf y})\mapsto ({\bf u}, {\bf v})=(s({\bf x},{\bf y}),t({\bf x},{\bf y}))\in \mathcal{X}\times \mathcal{X},$ will be called a {\em Yang-Baxter map} if it satisfies %the {\em Yang-Baxter relation}
\begin{align} \label{YANG_BAXTER}
R_{12}\circ R_{13}\circ R_{23}= R_{23}\circ R_{13}\circ R_{12},
\end{align}
where $R_{ij}$ $i\neq j\in\{1,2,3\},$ denotes the maps that act as  $R$ on the $i-$th and the $j-$th factor of $\mathcal{X}\times \mathcal{X}\times \mathcal{X},$ and as identity to the remaining factor.
\end{definition}
%The first instances of Yang-Baxter maps  appeared in \cite{Sklyanin:1988,Drinfeld:1992}. 
 Note that the term {\em Yang-Baxter maps} was introduced in \cite{Bukhshtaber:1998,Veselov:20031}. For recent developments on Yang-Baxter maps we refer to \cite{Kass2,Buchstaber_2020,Doikou_2021,Kassotakis:2:2021,Kassotakis:2022b,KONSTANTINOURIZOS2024,KASSOTAKIS:2025,Vincent:2025}.

%{\color{blue}
%\begin{definition}[Birational maps]
%An invertible   map $R: \mathcal{X} \times \mathcal{X}\ni({\bf x},{\bf y})\mapsto ({\bf u},{\bf v}) \in \mathcal{X} \times \mathcal{X}$   will be called %{\em birational}, if both the map $R$ and its inverse $R^{-1}: \mathcal{X} \times \mathcal{X}\ni({\bf u},{\bf v})\mapsto ({\bf x},{\bf y}) \in \mathcal{X} %\times \mathcal{X},$ are rational maps.
%\end{definition}
%}

\begin{definition}[Birational maps]
An invertible   map $R: \mathcal{X} \times \mathcal{X}\ni({\bf x},{\bf y})\mapsto ({\bf u},{\bf v}) \in \mathcal{X} \times \mathcal{X}$   will be called {\em birational}, if both the map $R$ and its inverse $R^{-1}: \mathcal{X} \times \mathcal{X}\ni({\bf u},{\bf v})\mapsto ({\bf x},{\bf y}) \in \mathcal{X} \times \mathcal{X},$ are rational maps.
\end{definition}

\begin{definition}[Quadrirational maps and their companion maps\footnote{
In \cite{Kakei2010} the companion maps are referred to as {\em solitonic maps}.
} \cite{Etingof_1999,ABS:YB}]
A map $R: \mathcal{X} \times \mathcal{X}\ni({\bf x},{\bf y})\mapsto ({\bf u},{\bf v}) \in \mathcal{X} \times \mathcal{X}$   will be called {\em quadrirational}, if both the map $R$ and the so-called {\em companion map} $R^c: \mathcal{X} \times \mathcal{X}\ni({\bf x},{\bf v})\mapsto ({\bf u},{\bf y}) \in \mathcal{X} \times \mathcal{X},$ are birational maps.
\end{definition}

In order to visualize quadrirational maps, we can assign the elements of the set  $\mathcal{X}$ on edges of elementary quads of the $\mathbb{Z}^2$ graph (see Figure \ref{fig1}). This representation has its origins in the theory of discrete integrable systems on quad-graphs \cite{Nijhoff:2002,bs:2002,ABS:YB}. In addition this representations  turns  convenient in the proof of the following proposition. This proposition although it  concerns quadrirational maps some of its points remain true in the non-quadrirational case. 

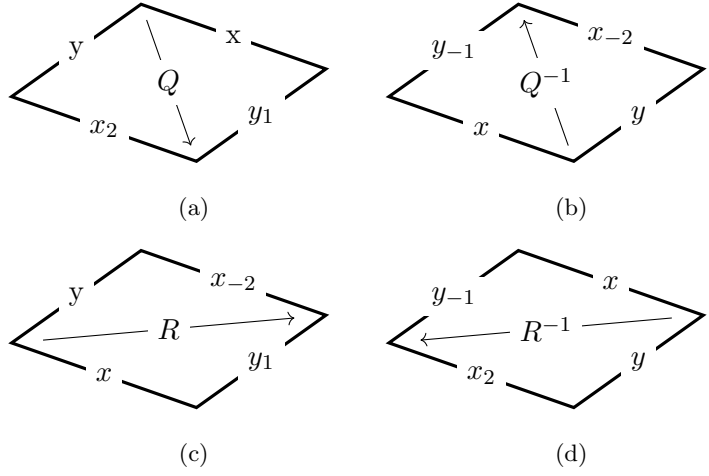
\begin{figure}[htb]
\begin{minipage}{.3\textwidth}
\begin{tikzpicture}[
		tdplot_main_coords,
		grid/.style={very thin,gray},
		axis/.style={->,blue,thick},
		%cube/.style={very thick,fill=red},
cube/.style={very thick},
		mynode/.style={circle,fill=green,minimum width=5pt,inner sep=3pt,outer sep=3pt}]
\draw[cube] (0,0,0)--(3,0,0)node[midway,fill=white]{y}--(3,3,0)node[midway,fill=white]{$x_2$}--(0,3,0)node[midway,fill=white]{$y_1$}--cyclenode[midway,fill=white]{x};
%\draw[cube] (0,0,0)--(2,0,0)--(2,2,0)--(0,2,0)--cycle;
\draw[->] (0.3,0.3,0)--(2.7,2.7,0)node[midway,fill=white]{$Q$};
\end{tikzpicture}
\captionsetup{font=footnotesize}%,calcwidth=0.8\columnwidth
\captionof*{figure}{(a)}
\end{minipage}
\begin{minipage}{.3\textwidth}
\begin{tikzpicture}[
		tdplot_main_coords,
		grid/.style={very thin,gray},
		axis/.style={->,blue,thick},
		%cube/.style={very thick,fill=red},
cube/.style={very thick},
		mynode/.style={circle,fill=green,minimum width=5pt,inner sep=3pt,outer sep=3pt}]
\draw[cube] (0,0,0)--(3,0,0)node[midway,fill=white]{$y_{-1}$}--(3,3,0)node[midway,fill=white]{$x$}--(0,3,0)node[midway,fill=white]{$y$}--cyclenode[midway,fill=white]{$x_{-2}$};
%\draw[cube] (0,0,0)--(2,0,0)--(2,2,0)--(0,2,0)--cycle;
\draw[<-] (0.3,0.3,0)--(2.7,2.7,0)node[midway,fill=white]{$Q^{-1}$};
\end{tikzpicture}
\captionsetup{font=footnotesize}%,calcwidth=0.8\columnwidth
\captionof*{figure}{(b)}
\end{minipage}\\ \vspace{0.3cm}
\begin{minipage}{.3\textwidth}
\begin{tikzpicture}[
		tdplot_main_coords,
		grid/.style={very thin,gray},
		axis/.style={->,blue,thick},
		%cube/.style={very thick,fill=red},
cube/.style={very thick},
		mynode/.style={circle,fill=green,minimum width=5pt,inner sep=3pt,outer sep=3pt}]
\draw[cube] (0,0,0)--(3,0,0)node[midway,fill=white]{y}--(3,3,0)node[midway,fill=white]{$x$}--(0,3,0)node[midway,fill=white]{$y_1$}--cycle node[midway,fill=white]{$x_{-2}$};
%\draw[cube] (0,0,0)--(2,0,0)--(2,2,0)--(0,2,0)--cycle;
\draw[->] (2.7,0.3,0)--(0.3,2.7,0)node[midway,fill=white]{$R$};
\end{tikzpicture}
\captionsetup{font=footnotesize}%,calcwidth=0.8\columnwidth
\captionof*{figure}{(c)}
\end{minipage}
\begin{minipage}{.3\textwidth}
\begin{tikzpicture}[
		tdplot_main_coords,
		grid/.style={very thin,gray},
		axis/.style={->,blue,thick},
		%cube/.style={very thick,fill=red},
cube/.style={very thick},
		mynode/.style={circle,fill=green,minimum width=5pt,inner sep=3pt,outer sep=3pt}]
\draw[cube] (0,0,0)--(3,0,0)node[midway,fill=white]{$y_{-1}$}--(3,3,0)node[midway,fill=white]{$x_{2}$}--(0,3,0)node[midway,fill=white]{$y$}--cycle node[midway,fill=white]{$x$};
%\draw[cube] (0,0,0)--(2,0,0)--(2,2,0)--(0,2,0)--cycle;
\draw[<-] (2.7,0.3,0)--(0.3,2.7,0)node[midway,fill=white]{$R^{-1}$};
\end{tikzpicture}
\captionsetup{font=footnotesize}%,calcwidth=0.8\columnwidth
\captionof*{figure}{(d)}
\end{minipage}
\caption{(a):Mapping $Q:({\bf x},{\bf y})\mapsto \left({\bf x_2},{\bf y_1}\right),$ assigned on an elementary quad of the $\mathbb{Z}^2$ graph. (b): The inverse map $Q^{-1}$ of $Q$. (c): The companion map $R,$ associated with mapping $Q$. (d): The inverse map $R^{-1}$ of the companion map $R$.}      \label{fig1}
\end{figure}

\begin{prop} \label{1st_prop}
Let $Q:({\bf x},{\bf y})\mapsto \left(f({\bf x},{\bf y}),g({\bf x},{\bf y})\right)$ be a quadrirational map. Then the following statements are equivalent
\begin{enumerate}
\item Mapping $Q$ is a $3D-$compatible map;
\item The inverse map $Q^{-1}$ of $Q$ is a $3D-$compatible map;
\item The companion map $R:=Q^c$ is a Yang-Baxter map;
\item The inverse $R^{-1}$ of the companion map $R$ is a Yang-Baxter map;
\item It holds
\begin{align*}
R_{12}\circ Q_{13}\circ Q_{23}=Q_{23}\circ Q_{13}\circ R_{12};
\end{align*}
\item It holds
\begin{align*}
R^{-1}_{12}\circ Q^{-1}_{13}\circ Q^{-1}_{23}=Q^{-1}_{23}\circ Q^{-1}_{13}\circ R^{-1}_{12};
\end{align*}
\item It holds
\begin{align*}
Q_{12}\circ Q_{13}\circ R^{-1}_{23}=R^{-1}_{23}\circ Q_{13}\circ Q_{12};
\end{align*}
\item It holds
\begin{align*}
Q^{-1}_{12}\circ Q^{-1}_{13}\circ R_{23}=R_{23}\circ Q^{-1}_{13}\circ Q^{-1}_{12};
\end{align*}
\item It holds
\begin{align*}
Q_{12}\circ R_{13}\circ Q^{-1}_{23}=Q^{-1}_{23}\circ R_{13}\circ Q_{12};
\end{align*}
\item It holds
\begin{align*}
Q^{-1}_{12}\circ R^{-1}_{13}\circ Q_{23}=Q_{23}\circ R^{-1}_{13}\circ Q^{-1}_{12}.
\end{align*}
\end{enumerate}
\end{prop}

\begin{proof}
The proof is presented in Appendix \ref{appa}.
\end{proof}

\begin{remark} \label{rema_000}
 Proposition \ref{1st_prop}  states that   mappings $Q,$ and $R$ and their inverses {\em entwine} with each other, that is they serve as solutions of the entwining Yang-Baxter equation. The entwining Yang-Baxter equation reads
 \begin{align} \label{eyb}
R_{12}^{(1)}\circ R_{13}^{(2)}\circ R_{23}^{(3)}= R_{23}^{(3)}\circ R_{13}^{(2)}\circ R_{12}^{(1)},
\end{align}
where the superscripts denote maps that might differ. For further discussions and for the first reported solutions of (\ref{eyb}) we refer to \cite{Kouloukas:2011}, see also \cite{Kassotakis:2019,Kels:2019II,Konstantinou-Rizos_2019} for further developments.
\end{remark}

\subsection{Idempotent $3D-$compatible maps}
In this article we focus on $3D-$compatible maps which are not quadrirational. Namely, we consider  M\"ob equivalence classes of maps (see (\ref{mob_eq}) below) that admit a representative of the form 
 %Here, we focus  on maps of the type
\begin{align}\label{2:0}
Q:(x,{\bf p};y,{\bf q}) \mapsto & (u,{\bf p};v,{\bf q}),
\end{align}
 where $x,y,u,v,p^i,q^i\in \mathbb{CP}^1$ and $u,v$ are implicitly defined by
\begin{align}  \label{2:1}
u=&v,& F(u,y,{\bf p},{\bf q})=&F(v,x,{\bf p},{\bf q}),
\end{align}
%or by
%\begin{align}  \label{2:1}
%x=&y,& F(u,y,{\bf p},{\bf q})=&F(v,x,{\bf p},{\bf q}),
%\end{align}
for a rational function $F.$ It is clear that the first relation in (\ref{2:1}), guarantees that mapping $Q$  is not invertible, hence not quadrirational. 
We are searching for maps $Q$ of the form (\ref{2:0}), (\ref{2:1}) which are rational and $3D-$compatible. 
%Demanding mapping $Q$ to be a rational   $3D-$compatible map, the function $F$ has to be chosen appropriately.  % The classification of all rational functions $F$ that correspond   to $3D-$compatible maps, will be given in our forthcoming studies.  
Here we investigate only the  case where the function $F$ is of multiplicative or additive separable form, that is
\begin{align*}
F(x,y,{\bf p},{\bf q}):=\frac{p^1x+p^2}{p^3x+p^4}\;\frac{q^1y+q^2}{q^3y+q^4},& &\mbox{or}& & F(x,y,{\bf p},{\bf q}):=&\frac{p^1x+p^2}{p^3x+p^4}+\frac{q^1y+q^2}{q^3y+q^4}.
\end{align*}
Studies on the relation of functions in separable form and discrete integrability, can be found in \cite{Kassotakis_2011,Kassotakis_2012,Kassotakis_2018,Kassotakis:2019}. Furthermore, we recall (see \cite{Papageorgiou:2010}), that change of variables implemented by families of bijections $\phi({\bf p}): \mathbb{CP}^1\rightarrow \mathbb{CP}^1, $ $ p^i\in \mathbb{CP}^1,$ respect the $3D-$compatibility of a map $Q$. Therefore, maps $Q$ and $\widetilde Q$ which are related by
\begin{align}\label{mob_eq}
\left(\phi({\bf p})\times \phi({\bf q})\right)\circ \widetilde Q=Q\circ \left(\phi({\bf p})\times \phi({\bf q})\right),
\end{align}
 will be called M\"ob equivalent.

We arrive at the following Theorem.

%We will need the following Lemma.

%\begin{lemma}
%The following  separable  functions:
%\begin{align*}
 % F(x,y,{\bf p},{\bf q}):=\frac{p^1x+p^2}{\delta x+1}\;\frac{q^1y+q^2}{\delta y+1}, \\
  %F(x,y,{\bf p},{\bf q}):=\frac{p^1x+p^2}{\delta x+1}+\frac{q^1y+q^2}{\delta y+1}
%\end{align*}
%define via (\ref{2:1}) rational maps.
%\end{lemma}

%\begin{lemma}
%In order   (\ref{2:0}) be a rational  $3D-$compatible map that admits a birational companion map,
%\end{lemma}

\begin{theorem} \label{theo1}
Any rational  $3D-$compatible map of type (\ref{2:0}), (\ref{2:1}), for separable function $F$ of the form $ F(x,y,{\bf p},{\bf q}):={\dis\frac{px-P}{\delta x-1}\;\frac{qy-Q}{\delta y-1}},$
  is M\"ob equivalent to exactly one of the following two maps:
\begin{align*}
u=&(P-Q)\frac{x}{x-y},&v=&(P-Q)\frac{y}{x-y},& (Q_I),\\
u=&\frac{q(x-y)+Q-P}{q-p},&v=&\frac{p(x-y)+Q-P}{q-p},& (Q_{II}),
\end{align*}
while for seperable function $F$ of the form $F(x,y,{\bf p},{\bf q}):={\dis\frac{px-P}{\delta x-1}+\frac{qy-Q}{\delta y-1}},$ it is M\"ob equivalent  to the map
\begin{align*}
u=&q\frac{x-y}{q-p},&v=&p\frac{x-y}{q-p}.& (Q_{III})
\end{align*}
\end{theorem}

\begin{proof}
Consider the pair of relations  (\ref{2:1})  for $ F(x,y,{\bf p},{\bf q}):={\dis\frac{p^1x+p^2}{p^3x+p^4}\;\frac{q^1y+q^2}{q^3y+q^4}}.$ Substituting $v=u,$ the second relation of (\ref{2:1}) becomes a quadratic polynomial equation in $u$. It is easy to see that the discriminant of this quadratic equation is a perfect square provided that it holds $p^3q^4=p^4q^3.$ This leads to $p^4=\delta p^3,$ $q^4=\delta q^3,$  so  $ F(x,y,{\bf p},{\bf q})= {\dis\frac{p^1x-p^2}{p^3(\delta x-1)}\;\frac{q^1y-q^2}{q^3(\delta y-1)}}$ and (\ref{2:1}) defines a rational map. This rational map, by setting $p:=p^1/p^3,P:=p^2/p^3,$ $q:=q^1/q^3,Q:=q^2/q^3,$ followed by the change of variables $(x,y)\mapsto\left(\frac{x-P}{\delta  x-p},\frac{y-Q}{\delta y-q}\right),$  and under the re-parametrization $(p,q)\mapsto (p+\delta P,q+\delta Q)$  becomes
$Q:(x,p,P;y,q,Q) \mapsto (u,p,P;v,q,Q),$
 where $\frac{p}{u-P}=\frac{q}{v-Q},$ $uy=vx.$
Re-defining the parameters as $(p,q,P,Q)\mapsto (1/p,1/q,P/p,Q/q)$  we get
\begin{align*}
Q:(x,p,P;y,q,Q) \mapsto & (u,p,P;v,q,Q),
\end{align*}
 where
\begin{align}  \label{q1}
p u-P=&q v-Q,& uy=&vx,
\end{align}
that is a rational map that after the scaling $(x,y)\mapsto (x/p,y/q)$ becomes exactly $Q_I.$

The proof that mapping $Q_I$ is $3D$-compatible follows by direct computation. Indeed, from
\begin{align*}
(Q_I)_{12}:(x,p,P;y,q,Q;z,r,R) \mapsto & (x_2,p,P;y_1,q,Q;z,r,R),\\
(Q_I)_{13}:(x,p,P;y,q,Q;z,r,R) \mapsto & (x_3,p,P;y,q,Q;z_1,r,R),\\
(Q_I)_{23}:(x,p,P;y,q,Q;z,r,R) \mapsto & (x,p,P;y_3,q,Q;z_2,r,R),
\end{align*}
where
\begin{align} \label{c1}
\begin{aligned}
x_2-P=&y_1-Q, & x_2y=&y_1x,\\
x_3-P=&z_1-R,& x_3z=&z_1x,\\
y_3-Q=&z_2-R,  & y_3z=&z_2y,
\end{aligned}
\end{align}
we have
\begin{align} \label{cc1}
\begin{aligned}
x_{23}-P=&y_{13}-Q,& x_{23}y_3=&y_{13}x_3,\\
x_{32}-P=&z_{12}-R, &x_{32}z_2=&z_{12}x_2,\\
y_{31}-Q=&z_{21}-R, &y_{31}z_1=&z_{21}y_1.
\end{aligned}
\end{align}
From (\ref{c1}) and (\ref{cc1}) it can be  shown that $x_{23}=x_{32},$ $y_{13}=y_{31},$ $z_{12}=z_{21},$ that results that mapping $Q_I$ is a $3D$-compatible map.

Considering the pair of relations  (\ref{2:1})  for $ F(x,y,{\bf p},{\bf q}):={\dis\frac{p^1x+p^2}{p^3x+p^4}+\frac{q^1y+q^2}{q^3y+q^4}},$ following the same analysis as above we obtain the rational map
\begin{align*}
Q:(x,p,P;y,q,Q) \mapsto & (u,p,P;v,q,Q),
\end{align*}
 where
\begin{align}  \label{q2}
pu-P=&qv-Q,& u+y=&v+x,
\end{align}
and the map
\begin{align*}
Q:(x,p,P;y,q,Q) \mapsto & (u,p,P;v,q,Q),
\end{align*}
 where
\begin{align}  \label{q3}
pu=&qv,& u+y=&v+x.
\end{align}
Mappings (\ref{q2}), respectively, (\ref{q3}) are $3D$-compatible rational maps  that coincide  with $Q_{II},$ respectively $Q_{III}.$

Note that the sets of singular points of $Q_I, Q_{II}$ and $Q_{II}$ respectively are
\begin{align*}
  \Sigma_{Q_I} =&\{(0,0),(\infty,\infty)\},& \Sigma_{Q_{II}} =&\{(\infty,\infty)^2\},& \Sigma_{Q_{III}} =&\{(\infty,\infty)\},
\end{align*}
that guarantees that these three maps are not M\"ob equivalent but serve as  representatives of three different equivalence classes of $3D-$ compatible maps.
\end{proof}

The following remarks are in order.
\begin{enumerate}[label=(\roman*)]
\item Mappings $Q_I, Q_{II}$ and $Q_{III},$ are idempotent, that is $Q_{I-III}^2=Q_{I-III}.$
\item Idempotency is being preserved under the equivalence relation (\ref{mob_eq}).  Indeed, Let $Q$ is an idempotent map and $ \widetilde Q:=\left(\phi^{-1}({\bf p})\times \phi^{-1}({\bf q})\right)\circ Q\circ \left(\phi({\bf p})\times \phi({\bf q})\right)$ a M\"ob equivalent map to $Q$. Then the idempotency property  $Q^2=Q$ implies  
    \begin{align*}
    \left(\phi^{-1}({\bf p})\times \phi^{-1}({\bf q})\right)\circ (Q^2-Q)\circ \left(\phi({\bf p})\times \phi({\bf q})\right)=0,
    \end{align*}
    or
    \begin{gather*}
    \left(\phi^{-1}({\bf p})\times \phi^{-1}({\bf q})\right)\circ Q^2\circ \left(\phi({\bf p})\times \phi({\bf q})\right)-\left(\phi^{-1}({\bf p})\times \phi^{-1}({\bf q})\right)\circ Q\circ \left(\phi({\bf p})\times \phi({\bf q})\right)=0,
    \end{gather*}
    or $\widetilde Q^2-\widetilde Q=0,$
    so $\widetilde Q$ is an idempotent map.
\item Mapping $Q_{II}$ can be obtained from $Q_I$ by a coalescence procedure. Indeed, starting with (\ref{q1}) that is M\"ob equivalent to $Q_I$, by setting
\begin{align*}
(x,y,u,v,p,q,P,Q)\mapsto \left(1+\epsilon x,1+\epsilon y,1+\epsilon u,1+\epsilon v,\epsilon^2+p,\epsilon^2+q,p+\epsilon P,q+\epsilon Q\right),
\end{align*}
and then sending $\epsilon \rightarrow 0,$ we obtain exactly $Q_{II}.$ The map $Q_{III}$ is obtained from $Q_{II}$ by setting $P=Q=0.$
\item  A map 
    \begin{align}\label{2:0_c}
\widehat Q:(u,{\bf p};v,{\bf q}) \mapsto & (x,{\bf p};y,{\bf q}),
\end{align}
 where $x,y,u,v,p^i,q^i\in \mathbb{CP}^1$ and $u,v$ are implicitly defined by
\begin{align}  \label{2:1_c}
x=&y,& F(u,y,{\bf p},{\bf q})=&F(v,x,{\bf p},{\bf q}),
\end{align}
%or by
%\begin{align}  \label{2:1}
%x=&y,& F(u,y,{\bf p},{\bf q})=&F(v,x,{\bf p},{\bf q}),
%\end{align}
for a rational function $F,$ will be called the complementary map of the map $Q$ which is defined by (\ref{2:0}),(\ref{2:1}).

Mappings $Q_{I-III}:(x,y)\mapsto (u,v),$ are accompanied by the {\em complementary maps} $\widehat Q_{I-III}:(u,v)\mapsto (x,y),$ where $x,y$ are rational  functions of $u,v$  respectively defined by (\ref{q1}), (\ref{q2}) and (\ref{q3}). It turns out that $Q_{I-III}\equiv \widehat Q_{I-III},$ hence the complementary maps $\widehat Q_{I-III},$ are also idempotent $3D-$ compatible maps. %but they do not serve as the inverses of mappings $Q_{I-III}.$ It holds though $Q_{I-III}\equiv \widehat Q_{I-III}.$  

\item  The companion maps $R_{I-III},$ of mappings $Q_{I-III}$ are birational maps that due Proposition \ref{1st_prop} are Yang-Baxter maps. Explicitly they read, $R_{I-III}:(u,p,P;y,q,Q) \mapsto  (x,p,P;v,q,Q),$ where
    \begin{align*}
x=&\frac{u y}{u+Q-P},&v=&u+Q-P,& (R_I),\\
x=&u+y-\frac{pu+Q-P}{q},&v=&\frac{pu+Q-P}{q},& (R_{II}),\\
x=&u+y-\frac{p}{q}u,&v=&\frac{p}{q}u.& (R_{III}).
\end{align*}
Also the inverse maps $R_{I-III}^{-1}$ are Yang-Baxter maps and they read $R_{I-III}^{-1}:(x,p,P;v,q,Q)\mapsto (u,p,P;y,q,Q),$ where
    \begin{align*}
u=&v+P-Q,&y=&\frac{v x}{v+P-Q},& (R_I^{-1}),\\
u=&\frac{qv+P-Q}{p},&y=&v+x-\frac{qv+P-Q}{p},& (R_{II}^{-1}),\\
u=&\frac{q}{p}v,&y=&v+x-\frac{q}{p}v.& (R_{III}^{-1}).
\end{align*}
\item %The companion maps $\widehat R_{I-III}$ of the complementary $3D-$ compatible maps $\widehat Q_{I-III},$ coincide with the inverses of the maps $R_{I-III}.$
     The companion maps $\widehat R_{I-III}$ of the complementary $3D-$ compatible maps $\widehat Q_{I-III},$ are related to the companion maps $R_{I-III},$ of $Q_{I-III}$ by:
     \begin{align*}
       \widehat R_{I-III}= & \tau\circ R_{I-III}\circ \tau,
     \end{align*}
     where $\tau$ the transposition map
     \begin{align*}
       \tau: (x,{\bf p};y,{\bf q}) \mapsto &(y,{\bf q}; x,{\bf p}).
     \end{align*}
\item Mappings $Q_{I-III},$ $\widehat Q_{I-III},$ $R_{I-III}$ and $R_{I-III}^{-1},$ satisfy those entwining relations of Proposition \ref{1st_prop} that $Q^{-1}_{I-III}$ do not participate.
\end{enumerate}

\section{Discrete integrable systems on the triangular lattice}\label{sec3}

\subsection{Compatible maps as difference systems}
Let $\mathcal{X}$ be any set. There is a natural identification of a map $Q:(x,{\bf p};y,{\bf q})\mapsto (x_2,{\bf p};y_1,{\bf q})=(f(x,y,{\bf p},{\bf q}),{\bf p};g(x,y,{\bf p},{\bf q}),{\bf q}),$ where $x,y,x_2,y_1,p^i,q^i\in \mathcal{X}$ to a difference system defined on the edges of the $\mathbb{Z}^2$ graph (see Figure \ref{fig_not0}).
This identification is  made by
\begin{align*}
\begin{aligned}
x:=&x_{m+1/2,n},& y:=&y_{m,n+1/2}, & x_1:=&x_{m+3/2,n}, & etc.\\
 x_2:=&x_{m+1/2,n+1},  &y_1:=&y_{m+1,n+1/2}, & y_2:=&y_{m,n+3/2},& etc.\\
 {\bf p}:=&{\bf p}_{m+1/2},&{\bf q}:=&{\bf q}_{n+1/2},& {\bf p}_1:=&{\bf p}_{m+3/2},& etc.
\end{aligned} & &  m,n\in \mathbb{Z},
\end{align*}
where with subscripts we denoted discrete shifts on the associated $\mathbb{Z}^2$ graph. Moreover, $x,y$  are considered as the functions $x,y:\mathbb{Z}^2\rightarrow \mathcal{X},$ while $p^i,q^i:\mathbb{Z}\rightarrow \mathcal{X}$ and the mapping $Q$ is in one-to-one correspondence with the difference system of equations
\begin{align}\label{dif_s}
(x_2,y_1)=(f(x,y,{\bf p},{\bf q}),g(x,y,{\bf p},{\bf q})).
\end{align}

\tdplotsetmaincoords{60}{200}
%\tdplotsetrotatedcoords{8}{8}{8} %<- rotate around (z,y,z)
\tdplotsetrotatedcoords{0}{20}{0} %<- rotate around (z,y,z)

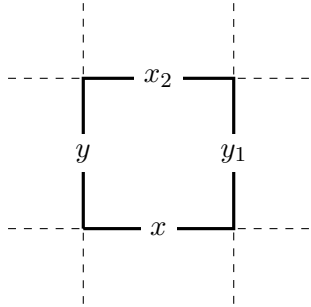
\begin{figure}[htb]
\begin{tikzpicture}[
		grid/.style={very thin,gray},
		axis/.style={->,blue,thick},
		%cube/.style={very thick,fill=red},
cube/.style={very thick},mynode/.style={circle,fill=green,minimum width=5pt,inner sep=3pt,outer sep=3pt}]
\draw[cube] (0,0)--(2,0)--(2,2)--(0,2)--(0,0);
\draw[dashed] (-1,0)--(0,0)--(0,-1);
\draw[dashed] (-1,2)--(0,2)--(0,3);
\draw[dashed] (2,3)--(2,2)--(3,2);
\draw[dashed] (2,-1)--(2,0)--(3,0);
\node[fill=white] at (1,0) {$x$};\node[fill=white] at (1,2) {$x_2$};
\node[fill=white] at (0,1) {$y$};\node[fill=white] at (2,1) {$y_1$};
\end{tikzpicture}
\caption{The difference system (\ref{dif_s}) assigned at the edges of the $\mathbb{Z}^2$ graph}\label{fig_not0}
\end{figure}

%\begin{tikzpicture}
%\coordinate (Origin)   at (0,0);
%\coordinate (XAxisMin) at (-3,0);
%\coordinate (XAxisMax) at (3,0);
%\coordinate (YAxisMin) at (0,-2);
%\coordinate (YAxisMax) at (0,3);
%\clip (-3.5,-3) rectangle (3.1,3.1);
%\begin{scope}[y=(60:1)]
%\foreach \x  [count=\j from 2] in {-8,-6,...,8}{% Two indices running over each

 % \draw[help lines,dashed]
 %   (\x,-8) -- (\x,8)
 %   (-8,\x) -- (8,\x)
 %   [rotate=60] (\x,-8) -- (\x,8)} ;
 %   \end{scope}

%\end{tikzpicture}

%\begin{figure}[htb]
%\begin{tikzpicture}[
%		grid/.style={very thin,gray},
%		axis/.style={->,blue,thick},
		%cube/.style={very thick,fill=red},
%cube/.style={very thick},mynode/.style={circle,fill=green,minimum width=5pt,inner sep=3pt,outer sep=3pt}]
%\draw[cube] (0,0)--(2,0)--(4,0);
%\draw[cube] (0,2)--(2,2)--(4,2);
%\draw[cube] (0,4)--(2,4)--(4,4);
%\draw[cube] (0,0)--(0,2)--(0,4);
%\draw[cube] (2,0)--(2,2)--(2,4);
%\draw[cube] (4,0)--(4,2)--(4,4);
%\node[fill=white] at (1,0) {$x$};\node[fill=white] at (1,2) {$x_2$};
%\node[fill=white] at (3,0) {$x_1$};\node[fill=white] at (3,2) {$x_{12}$};
%%\node[fill=white] at (0,1) {$y$};\node[fill=white] at (2,1) {$y_1$};
%%\node[fill=white] at (0,3) {$y_2$};\node[fill=white] at (2,3) {$y_{12}$};
%\end{tikzpicture}
%\end{figure} 

Provided that the underlying map $Q$ is $3D-$compatible, the associated difference systems can be extended as a compatible difference system on the   $\mathbb{Z}^N$ graph. This extension  on the $\mathbb{Z}^N$ graph is obtained by the identification
\begin{align} \label{notation1.1}
\begin{aligned}
x^{i}:=&x,& {\bf p}^i:=&{\bf p},&  x^{j}:=&y,& {\bf p}^j:=&{\bf q},& x^{i}_j:=&x_j,& x^{j}_i:=y_i,&
\end{aligned} & i\neq j \in \{1,\ldots, N\},
\end{align}
while the $3D$-compatibility guarantees the existence of such a system on the $\mathbb{Z}^N$ graph.  In the  notation (\ref{notation1.1}) the superscripts  represent the associated edges of the $\mathbb{Z}^N$ graph where the variables are assigned to. In addition, provided that (\ref{dif_s}) admits the symmetry
$f(y,x,{\bf q},{\bf p})=g(x,y,{\bf p},{\bf q}),$ the multidimensional extension of (\ref{dif_s}) can be represented in the following compact form
\begin{align*}
   x^{i}_j=&f(x^i,x^j, {\bf p}^{i}, {\bf p}^{j}), & i\neq j \in& \{1,\ldots, N\},
\end{align*}
where  $x^{i}_{jk}=x^{i}_{kj},$ $i\neq j\neq k\neq i\in \{1,\ldots, N\},$
%\begin{align} \label{mult-com}
% x^{i}_{jk}=&x^{i}_{kj}, & i\neq j\neq k\neq i\in \{1,\ldots, N\},
%\end{align}
holds.
\subsection{The difference systems associated with $Q_{I-III}, \widehat Q_{I-III}$ are multidimensional compatible}

Following the identification of the previous Section, in Table \ref{table1} we present the difference systems in edge variables associated with the maps $Q_{{I}-{III}}$ of Theorem \ref{theo1} and the difference systems associated with the complementary maps $\widehat Q_{{I}-{III}}.$
\begin{table}[h]
\begin{tabular}{c|ll|ll}
  \hline
  {Difference}\\ {system} & {}& {}& {}&{}\\ \hline
  $Q_I$ & $
          x_2-P=y_1-Q$,& $x_2y=y_1x,$&
          $T_2(x-P)=T_1(y-Q)$,& $\frac{T_2}{id}(x)=\frac{T_1}{id}(y)$ \\
  $Q_{II}$ & $ p x_2-P=q y_1-Q$,& $x_2+y=y_1+x $&$ T_2(px-P)=T_1(qy-Q),$& $(T_2-id)(x)=(T_1-id)(y) $\\
  $Q_{III}$ &  $   p x_2=q y_1$,& $x_2+y=y_1+x$& $  T_2( p x)=T_1( q y)$,& $(T_2-id)(x)=(T_1-id)(y)$ \\
 $\widehat Q_I$ &  $
          x-P=y-Q$,& $x_2y=y_1x$ & $
          id(x-P)=id(y-Q)$,& $\frac{T_2}{id}(x)=\frac{T_1}{id}(y)$ \\
  $\widehat Q_{II}$ & $ p x-P=q y-Q$,& $x_2+y=y_1+x $& $ id(p x-P)=id(q y-Q)$,& $(T_2-id)(x)=(T_1-id)(y)$\\
  $\widehat Q_{III}$ & $   p x=q y$,& $x_2+y=y_1+x$& $   id(p x)=id(q y)$,& $(T_2-id)(x)=(T_1-id)(y)$ \\
    \hline
\end{tabular}
\caption{The difference systems $Q_{{I}-{III}}, \widehat Q_{{I}-{III}}$ in conservation form, where~$T_i,$ $i=1,2$ the shift operators i.e. $T_i (f):=f_i$}\label{table1}
\end{table}
The  difference systems associated with $Q_{{I}-{III}}$ and  $\widehat Q_{{I}-{III}}$ should be treated separably since when they are imposed on the same elementary quad of the $\mathbb{Z}^2$  graph, leads to incompatibility.

\subsubsection{Multidimensional compatibility}
The $3D-$compatibility of $Q_{{I}-{III}},$ and  $\widehat Q_{{I}-{III}},$ results the multidimensional compatible extension of the associated difference systems in edge variables. In detail, the multidimensional extension of $Q_{{I}-{III}},$ and the associated multidimensional compatibility formulas respectively read
\begin{align*}
(Q_I)&& x^i_j=&(P^i-P^j)\frac{x^i}{x^i-x^j},&&x^i_{jk}=\frac{(P^i-P^j)(P^i-P^k)(x^j-x^k)x^i}{(P^i-P^j)x^ix^j-(P^i-P^k)x^ix^k+(P^j-P^k)x^jx^k}, \\
(Q_{II})&&x^i_j=&\frac{p^j(x^j-x^i)+P^j-P^i}{p^j-p^i},&&\begin{multlined} x^i_{jk}=\frac{(P^i-P^j)p^k}{(p^i-p^j)(p^j-p^k)}-\frac{(P^i-P^k)p^j}{(p^i-p^k)(p^j-p^k)}\\ -\frac{p^jp^k\left((p^j-p^k)x^i+(p^k-p^i)x^j+(p^i-p^j)x^k\right)}{(p^i-p^j)(p^j-p^k)(p^k-p^i)},\end{multlined}\\
(Q_{III})&&x^i_j=&\frac{p^j(x^j-x^i)}{p^j-p^i},&& x^i_{jk}=-\frac{p^jp^k\left((p^j-p^k)x^i+(p^k-p^i)x^j+(p^i-p^j)x^k\right)}{(p^i-p^j)(p^j-p^k)(p^k-p^i)},
\end{align*}
where it is clear that the expressions $x^{i}_{jk}$ are invariant under the interchange $j\leftrightarrow k$ so   $x^{i}_{jk}=x^{i}_{kj},$ $i\neq j\neq k\neq i\in \{1,\ldots, N\},$ holds.
While for the multidimensional extension of  $\widehat Q_{{I}-{III}},$ we have
\begin{align*}
(\widehat Q_I)&& x^i_{-j}=&(P^i-P^j)\frac{x^i}{x^i-x^j},&&x^i_{-j-k}=\frac{(P^i-P^j)(P^i-P^k)(x^j-x^k)x^i}{(P^i-P^j)x^ix^j-(P^i-P^k)x^ix^k+(P^j-P^k)x^jx^k}, \\
(\widehat Q_{II})&&x^i_{-j}=&\frac{p^j(x^j-x^i)+P^j-P^i}{p^j-p^i},&&\begin{multlined} x^i_{-j-k}=\frac{(P^i-P^j)p^k}{(p^i-p^j)(p^j-p^k)}-\frac{(P^i-P^k)p^j}{(p^i-p^k)(p^j-p^k)}\\ -\frac{p^jp^k\left((p^j-p^k)x^i+(p^k-p^i)x^j+(p^i-p^j)x^k\right)}{(p^i-p^j)(p^j-p^k)(p^k-p^i)},\end{multlined}\\
(\widehat Q_{III})&&x^i_{-j}=&\frac{p^j(x^j-x^i)}{p^j-p^i},&& x^i_{-j-k}=-\frac{p^jp^k\left((p^j-p^k)x^i+(p^k-p^i)x^j+(p^i-p^j)x^k\right)}{(p^i-p^j)(p^j-p^k)(p^k-p^i)},
\end{align*}
and $x^{i}_{-j-k}=x^{i}_{kj},$ $i\neq j\neq k\neq i\in \{1,\ldots, N\},$ holds.
%{\color{blue} 3d-compatible formulas for maps and complementary maps $x^i_{-j-k}$\\
%initial value problem}
%\input{fig_init}

\subsection{Idempotent compatible maps and discrete integrable systems on the triangular lattice }
In the previous Section, it was discussed the correspondence of maps with difference systems in edge variables. Here, we exploit this correspondence while focusing on the idempotent $3D-$compatible  maps of Section \ref{sec2}. In detail, we show that to the difference systems in edge variables corresponding to $Q_{I-III}$ and to the complementary systems $\widehat Q_{I-III}$, are associated difference systems in vertex variables defined on a triangular stencil. In that respect, the difference systems in vertex variables associated with $Q_{I-III}$ are defined on the ``black" triangles  (see Figure \ref{Figure2} (b)) of the triangular lattice $Q(A_2),$ while
 the difference systems  associated with $\hat Q_{I-III}$ are defined on the ``white" triangles of the  $Q(A_2)$ lattice. Note that the $Q(A_2)$ it is defined by
 $
 Q(A_2)=\{(l,m,n)\in\mathbb{Z}^3:l+m+n=0\}.
 $
 We consider the functions $\phi: Q(A_2)\rightarrow \mathbb{CP}^1,$ $\phi: (l,m,n)\mapsto \phi_{l,m,n},$ where $\phi_{l,m,n}:=T_{\hat 1}^lT_{\hat 2}^mT_{\hat 3}^n (\phi_{0,0,0})$ and $T_{\hat 1},T_{\hat 2}$ and $T_{\hat 3}$ the elementary shift operators. Note that there is $T_{\hat 1}T_{\hat 2}T_{\hat 3}=id,$ so $\phi_{l+a,m+a,n+a}=\phi_{l,m,n},$ $a\in \mathbb{Z}.$ In order to eliminate this redundancy, we choose the independent shifts $T_1:=T_{\hat 1},T_2:=T_{\hat 3}^{-1}$ so $\phi_{l',m'}:=\phi_{l,0,-n}.$ Furthermore, it is convenient to adopt the notation
 \begin{align}\label{not2}
 \phi_{l'-1,m'}:=&\phi_{-1},&\phi_{l',m'}:=&\phi,&\phi_{l'+1,m'}:=&\phi_1,&\phi_{l',m'+1}:=&\phi_2,&\phi_{l'+2,m'+1}:=&\phi_{112},& etc.
 \end{align}

\tdplotsetmaincoords{60}{200}
%\tdplotsetrotatedcoords{8}{8}{8} %<- rotate around (z,y,z)
\tdplotsetrotatedcoords{0}{20}{0} %<- rotate around (z,y,z)

\begin{figure}[htb]
\begin{minipage}{.4\textwidth}
\adjustbox{scale=0.5,center}{\begin{tikzpicture}[
		tdplot_main_coords,
		grid/.style={very thin,gray},
		axis/.style={->,blue,thick},
		%cube/.style={very thick,fill=red},
cube/.style={very thick},
		mynode/.style={circle,fill=green,minimum width=3pt,inner sep=5pt,outer sep=5pt},
  extended line/.style={shorten >=-#1,shorten <=-#1},
  extended line/.default=1cm]
%\draw[cube] (0,0,0)--(3,0,0)--(3,3,0)--(0,3,0)--cycle;
%\draw[cube] (3,0,0)--(6,0,0)--(6,3,0)--(3,3,0);
%\draw[cube] (3,3,0)--(3,6,0)--(6,6,0)--(6,3,0);
%\draw[cube] (3,6,0)--(0,6,0)--(0,3,0);
%\draw[cube] (6,3,0)--(3,6,0) (6,0,0)--(3,3,0)--(0,6,0) (3,0,0)--(0,3,0);
\draw[fill] (0,0,0) circle (3pt);
\node[fill=white] at (3,0,0) {$\phi$}; \draw[fill] (3,0,0) circle (3pt);
\node[fill=white] at (3,3,0) {$\phi$}; \draw[fill] (3,3,0) circle (3pt);
\node[fill=white] at (0,3,0) {$\phi$}; \draw[fill] (0,3,0) circle (3pt);
\node[fill=white] at (6,0,0) {$\phi$}; \draw[fill] (6,0,0) circle (3pt);
\node[fill=white] at (6,3,0) {$\phi$}; \draw[fill] (6,3,0) circle (3pt);
\node[fill=white] at (0,6,0) {$\phi$}; \draw[fill] (0,6,0) circle (3pt);
\node[fill=white] at (3,6,0) {$\phi$}; \draw[fill] (3,6,0) circle (3pt);
\node[fill=white] at (6,6,0) {$\phi$}; \draw[fill] (6,6,0) circle (3pt);
%%%%%%m==
\draw[dashed,extended line=0.5cm] (0,9,0)--(0,-3,0)node[right]{$n=1$};
\draw[dashed] (3,9,0)--(3,-3,0)node[right]{$n=0$};
\draw[dashed] (6,9,0)--(6,-3,0)node[right]{$n=-1$};
%%%%%l==
\draw[dashed] (9,6,0)--(-3,6,0)node[right]{$l=-1$};
\draw[dashed] (9,3,0)--(-3,3,0)node[right]{$l=0$};
\draw[dashed,extended line=-0.5cm] (9,0,0)--(-3,0,0)node[right]{$l=1$};
%%m=
\draw[dashed] (-3,9,0)--(9,-3,0)node[right]{$m=0$};
\draw[dashed,extended line=0.5cm,fill=white] (0,9,0)--(9,0,0)node[left]{$m=-1$};
\draw[dashed,extended line=1.5cm,fill=white] (0,3,0)--(3,0,0);
\end{tikzpicture}}
\captionsetup{font=footnotesize}%,calcwidth=0.8\columnwidth
\captionof*{figure}{(a) The $Q(A_2)$ lattice}
\end{minipage}
\begin{minipage}{.4\textwidth}
\adjustbox{scale=0.5,center}{\begin{tikzpicture}[
		tdplot_main_coords,
		grid/.style={very thin,gray},
		axis/.style={->,blue,thick},
		%cube/.style={very thick,fill=red},
cube/.style={thick,dashed,opacity=0.2,draw=white},
		mynode/.style={circle,fill=green,minimum width=5pt,inner sep=3pt,outer sep=3pt}]
\draw[cube] (0,0,0)--(3,0,0)--(3,3,0)--(0,3,0)--cycle;
\draw[cube] (3,0,0)--(6,0,0)--(6,3,0)--(3,3,0);
\draw[cube] (3,3,0)--(3,6,0)--(6,6,0)--(6,3,0);
\draw[cube] (3,6,0)--(0,6,0)--(0,3,0);
\draw[cube] (6,3,0)--(3,6,0) (6,0,0)--(3,3,0)--(0,6,0) (3,0,0)--(0,3,0);
\node[fill=white] at (0,0,0) {$\phi_{12}$};
\node[fill=white] at (3,0,0) {$\phi_{2}$};
\node[fill=white] at (3,3,0) {$\phi$};
\node[fill=white] at (0,3,0) {$\phi_{1}$};
\node[fill=white] at (6,-3,0) {};
\node[fill=white] at (6,0,0) {$\phi_{-12}$};
\node[fill=white] at (6,3,0) {$\phi_{-1}$};
\node[fill=white] at (0,6,0) {$\phi_{1-2}$};
\node[fill=white] at (3,6,0) {$\phi_{-2}$};
\node[fill=white] at (6,6,0) {$\phi_{-1-2}$};

\draw[fill=gray,fill opacity=0.2,draw=white] (3,6,0)--(6,3,0)--(3,3,0)--cycle;
\draw[fill=gray,fill opacity=0.2,draw=white] (3,3,0)--(6,0,0)--(3,0,0)--cycle;
\draw[fill=gray,fill opacity=0.2,draw=white] (0,6,0)--(3,3,0)--(0,3,0)--cycle;
\draw[fill=gray,fill opacity=0.2,draw=white] (0,3,0)--(3,0,0)--(0,0,0)--cycle;
\draw[fill=gray,fill opacity=0.2,draw=white] (6,-3,0)--(3,0,0)--(3,-3,0)--cycle;
\end{tikzpicture}}
\captionsetup{font=footnotesize}%,calcwidth=0.8\columnwidth
\captionof*{figure}{(b) The function $\phi$ assigned at the $Q(A_2)$ lattice }\label{fig22}
\end{minipage}
\caption{{}}\label{Figure2}
\end{figure}

Consider the difference system $Q_I$ that in conservation form (see Table \ref{table1}) reads
\begin{align} \label{q1_e1}
T_2(x-P)=&T_1(y-Q),\\ \label{q1_e2}
\frac{T_2 (x)}{x}=&\frac{T_1 (y)}{y}.
\end{align}
The conservation relation (\ref{q1_e1}) guarantees the existence of a function $\phi$ such that
\begin{align}\label{potq11}
x-P=&T_1(\phi),&y-Q=&T_2(\phi).
\end{align} In terms of $\phi,$ (\ref{q1_e2}) reads
\begin{align}\label{burger}
(\phi_{12}+P)(\phi_2+Q)=&(\phi_{12}+Q)(\phi_1+P).
\end{align}
On the other hand, the conservation relation (\ref{q1_e2}) guarantees the existence of a potential function $\chi$ such that
\begin{align}\label{potq12}
x=&\chi/T_1(\chi),&y=&\chi/T_2(\chi).
\end{align}
 In terms of the potential  $\chi$ (\ref{q1_e1}) reads
\begin{align}\label{burger_l}
\chi_2-\chi_1=&(P-Q)\chi_{12}.
\end{align}
By eliminating $x,y$ from (\ref{potq11}) and (\ref{potq12}) we obtain that (\ref{burger}) and (\ref{burger_l}) are related by the substitution
\begin{align}\label{sub1}
\frac{\chi}{\chi_1}-P=&\phi_1,& \frac{\chi}{\chi_2}-Q=&\phi_2,
\end{align}
hence (\ref{burger}) is a linearizable equation. Equation (\ref{burger}) serves as a discretization of the Burgers equation and it was first derived in \cite{Levi:1983}, c.f. \cite{Zhang:2021,Zhang:2022}. Both the discrete Burgers equation  (\ref{burger}) and its associated linear equation  (\ref{burger_l}) are defined on the black triangles of the $Q(A2)$ lattice, see Figure \ref{Figure2} $(b)$.  Working similarly, from the complementary  difference system $\widehat Q_I$ we obtain the following two difference systems defined on the vertices of the white triangles  of a $Q(A2)$ lattice
\begin{align*}
(\phi+P)(\phi_1+Q)=&(\phi+Q)(\phi_2+P), & \chi_1-\chi_2=&(P-Q)\chi,
\end{align*}
which are related by
\begin{align*}
\frac{\chi_{1}}{\chi}-P=&\phi,& \frac{\chi_{2}}{\chi}-Q=&\phi.
\end{align*}

%In the following Table we present the difference systems in vertex variables  defined on the black and the white triangles of the $Q(A2)$ lattice, which %correspond to the $Q_{I-III}$ and the $Q_{I-III}$ edge systems.

In the following Table we present the difference systems in vertex variables  defined on the black  triangles of the $Q(A2)$ lattice, which correspond to the $Q_{I-III}$  edge systems.

\begin{table}[h]
  \centering
  \begin{tabular}{|c|c|c|}
    \hline
    Difference system & Potentials & Triangular equations \\ \hline
    $Q_I$ & $\begin{array}{cc}
              x=\phi_1+P & y=\phi_2+Q\\
              x=\frac{\chi}{\chi_1}, & y=\frac{\chi}{\chi_2}
              \end{array} $ & $\begin{array}{c}(\phi_{12}+P)(\phi_2+Q)=(\phi_{12}+Q)(\phi_1+P)\\
                                     \chi_2-\chi_1=(P-Q)\chi_{12}\end{array}$ \\ [3mm]
    $Q_{II}$ & $\begin{array}{cc}
                 x=\frac{\phi_1+P}{p}, & y=\frac{\phi_2+Q}{q}\\
                 x=\chi_1-\chi, & y=\chi_2-\chi
                  \end{array}$ & $\begin{array}{c}
                       \left(\frac{1}{p}-\frac{1}{q}\right)\phi_{12}=\frac{1}{p}\phi_1-\frac{1}{q}\phi_2,\\
                       (p-q)\chi_{12}=q\chi_1-p\chi_2+P-Q
                       \end{array}$ \\ [3mm]
    $Q_{III}$ & $\begin{array}{cc}
                 x=\frac{\phi_1}{p}, & y=\frac{\phi_2}{q}\\
                 x=\chi_1-\chi, & y=\chi_2-\chi
                  \end{array} $ & $\begin{array}{c}
                  \left(\frac{1}{p}-\frac{1}{q}\right)\phi_{12}=\frac{1}{p}\phi_1-\frac{1}{q}\phi_2\\
                  (p-q)\chi_{12}=q\chi_1-p\chi_2
                  \end{array}$ \\
    \hline
  \end{tabular}
  \caption{Potentials and associated difference systems in edge variables corresponding to $Q_{I-III}$}\label{idolons}
\end{table}

%\begin{table}[h]
 % \centering
  %\begin{tabular}{|c|c|c|}
  %  \hline
    %Difference system & 1st triangular equation & 2nd triangular equation \\ \hline
    %$Q_I$ & $(\phi_{12}+P)(\phi_2+Q)=(\phi_{12}+Q)(\phi_1+P)$ & $\chi_2-\chi_1=(P-Q)\chi_{12}$ \\ [3mm]
    %$Q_{II}$ & $\left(\frac{1}{p}-\frac{1}{q}\right)\phi_{12}=\frac{1}{p}\phi_1-\frac{1}{q}\phi_2$ & $(p-q)\chi_{12}=q\chi_1-p\chi_2+P-Q$ \\ [3mm]
    %$Q_{III}$ & $\left(\frac{1}{p}-\frac{1}{q}\right)\phi_{12}=\frac{1}{p}\phi_1-\frac{1}{q}\phi_2$ & $(p-q)\chi_{12}=q\chi_1-p\chi_2$ \\
    %\hline
  %\end{tabular}
  %\caption{..}\label{idolons}
%\end{table}

\subsubsection{An initial value problem on the $Q(A2)$ lattice}
A well posed initial value problem for the difference systems in vertex variables associated with the idempotent maps $Q_{I-III}$ is presented in Figure \ref{figfin}.

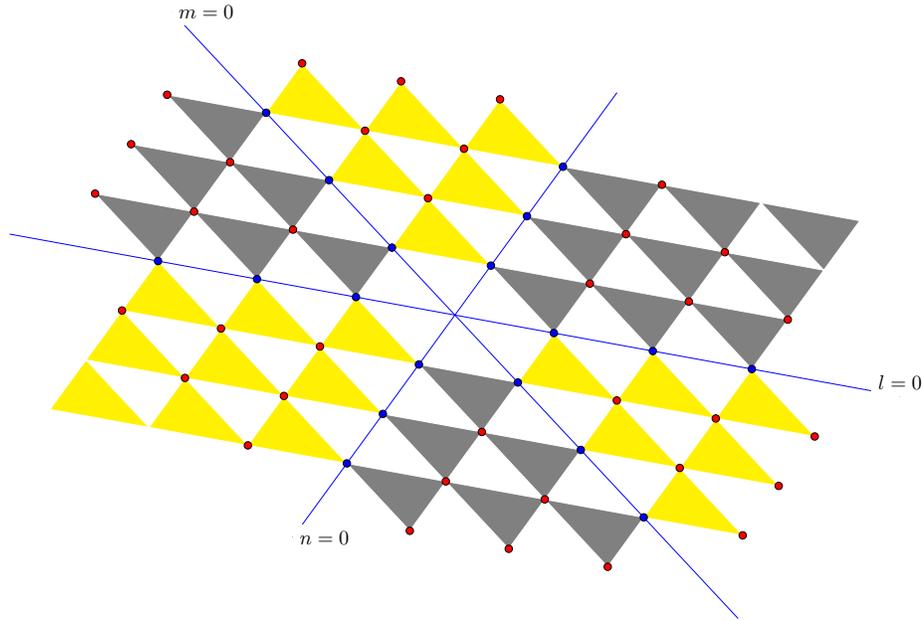
\begin{figure}[htb]\adjustbox{scale=0.7,center}{
\begin{tikzpicture}[tdplot_main_coords,		grid/.style={very thin,gray},
		axis/.style={->,blue,thick},
		%cube/.style={very thick,fill=red},
cube/.style={very thin,dashed},mynode/.style={circle,fill=green,minimum width=5pt,inner sep=3pt,outer sep=3pt}]
%\draw[cube] (0,0)--(2,0)--(2,2)--(0,2)--(0,0);
%\draw[cube] (0,0)--(0,-2)--(-2,-2)--(-2,-4)--(-4,-4)--(-4,-2)--(-2,-2)--(-2,0)--(0,0);
%\draw[cube] (0,-2)--(2,-2)--(4,-2)--(4,0)--(4,2)--(4,4)--(4,2)--(4,0)--(2,0);
%\draw[cube] (2,2)--(4,2) (2,2)--(2,4)--(4,4) (2,0)--(2,-2) (2,4)--(0,4)--(0,2)--(-2,2)--(0,2)  (0,4)--(-2,4)--(-2,2)--(-2,0) %(-2,4)--(-4,4)--(-4,2)--(-4,0)--(-4,-2) (-4,0)--(-2,0) (-4,2)--(-2,2) (-2,-4)--(0,-4)--(2,-4)--(4,-4)--(4,-2) (2,-4)--(2,-2) (0,-4)--(0,-2);

%\draw[cube] (-8,0)--(8,0) (-8,2)--(8,2) (-8,4)--(8,4) (-8,6)--(8,6) (-8,8)--(8,8)  (-8,-2)--(8,-2) (-8,-4)--(8,-4) (-8,-6)--(8,-6) (-8,-8)--(8,-8);
%\draw[cube] (-8,8)--(-8,-8) (-6,8)--(-6,-8) (-4,8)--(-4,-8) (-2,8)--(-2,-8) (-0,8)--(-0,-8) (2,8)--(2,-8) (4,8)--(4,-8) (6,8)--(6,-8) (8,8)--(8,-8);
%\draw[cube] (-8,8)--(8,-8)  (-6,8)--(8,-6) (-4,8)--(8,-4) (-2,8)--(8,-2) (0,8)--(8,0) (2,8)--(8,2) (4,8)--(8,4) (6,8)--(8,6);
%\draw[cube] (-8,-6)--(-6,-8) (-8,-4)--(-4,-8)  (-8,-2)--(-2,-8) (-8,0)--(0,-8) (-8,2)--(2,-8) (-8,4)--(4,-8) (-8,6)--(6,-8);
\draw[solid,blue] (0,-9)--(0,9) (-9,0)--(9,0);\draw[blue] (-9,9)--(9,-9);  \node[fill=white, above] at (-9,0) {$l=0$};  \node[fill=white, right] at (0,9) {$n=0$}; \node[fill=white, right] at (9,-9) {$m=0$};
%4rth-champer
\draw[fill=yellow,fill opacity=0.2,draw=white] (0,2)--(2,2)--(2,0)--cycle (0,4)--(2,4)--(2,2)--cycle (0,6)--(2,6)--(2,4)--cycle    (2,2)--(4,2)--(4,0)--cycle (2,4)--(4,4)--(4,2)--cycle (2,6)--(4,6)--(4,4)--cycle             (4,2)--(6,2)--(6,0)--cycle (4,4)--(6,4)--(6,2)--cycle (4,6)--(6,6)--(6,4)--cycle;
%1st-champer
\draw[fill=gray,fill opacity=0.2,draw=white] (0,-2)--(-2,-2)--(-2,0)--cycle (0,-4)--(-2,-4)--(-2,-2)--cycle (0,-6)--(-2,-6)--(-2,-4)--cycle    (-2,-2)--(-4,-2)--(-4,0)--cycle (-2,-4)--(-4,-4)--(-4,-2)--cycle (-2,-6)--(-4,-6)--(-4,-4)--cycle             (-4,-2)--(-6,-2)--(-6,0)--cycle (-4,-4)--(-6,-4)--(-6,-2)--cycle (-4,-6)--(-6,-6)--(-6,-4)--cycle;
%4th-champer
\draw[fill=gray,fill opacity=0.2,draw=white] (0,2)--(-2,4)--(-2,2)--cycle    (0,4)--(-2,6)--(-2,4)--cycle (0,6)--(-2,8)--(-2,6)--cycle          (-2,4)--(-4,6)--(-4,4)--cycle (-2,6)--(-4,8)--(-4,6)--cycle (-4,6)--(-6,8)--(-6,6)--cycle;
%6th=chamber
\draw[fill=yellow,fill opacity=0.2,draw=white] (-2,0)--(-2,2)--(-4,2)--cycle  (-4,0)--(-4,2)--(-6,2)--cycle (-6,0)--(-6,2)--(-8,2)--cycle        (-4,2)--(-4,4)--(-6,4)--cycle (-6,2)--(-6,4)--(-8,4)--cycle          (-6,4)--(-6,6)--(-8,6)--cycle;
%%2nd-chamber
\draw[fill=gray,fill opacity=0.2,draw=white] (2,0)--(2,-2)--(4,-2)--cycle  (4,0)--(4,-2)--(6,-2)--cycle  (6,0)--(6,-2)--(8,-2)--cycle    (4,-2)--(4,-4)--(6,-4) (6,-4)--(6,-6)--(8,-6)--cycle       (6,-2)--(6,-4)--(8,-4)--cycle;

\draw[fill=yellow,fill opacity=0.2,draw=white] (0,-2)--(2,-2)--(2,-4)--cycle    (0,-4)--(2,-4)--(2,-6)--cycle     (0,-6)--(2,-6)--(2,-8)--cycle     (2,-4)--(4,-4)--(4,-6)--cycle     (2,-6)--(4,-6)--(4,-8)--cycle    (4,-6)--(6,-6)--(6,-8)--cycle;    

\draw[fill=blue,opacity=0.6] (2,0) circle (2pt); \draw[fill=blue,opacity=0.6] (4,0) circle (2pt); \draw[fill=blue,opacity=0.6] (6,0) circle (2pt); \draw[fill=blue,opacity=0.6] (-2,0) circle (2pt); \draw[fill=blue,opacity=0.6] (-4,0) circle (2pt); \draw[fill=blue,opacity=0.6] (-6,0) circle (2pt); 
\draw[fill=blue,opacity=0.6] (0,2) circle (2pt); \draw[fill=blue,opacity=0.6] (0,4) circle (2pt); \draw[fill=blue,opacity=0.6] (0,6) circle (2pt); 
\draw[fill=blue,opacity=0.6] (0,-2) circle (2pt); \draw[fill=blue,opacity=0.6] (0,-4) circle (2pt); \draw[fill=blue,opacity=0.6] (0,-6) circle (2pt); 
\draw[fill=blue,opacity=0.6] (2,-2) circle (2pt); \draw[fill=blue,opacity=0.6] (4,-4) circle (2pt); \draw[fill=blue,opacity=0.6] (6,-6) circle (2pt);
\draw[fill=blue,opacity=0.6] (-2,2) circle (2pt); \draw[fill=blue,opacity=0.6] (-4,4) circle (2pt); \draw[fill=blue,opacity=0.6] (-6,6) circle (2pt);

\draw[fill=red,opacity=0.6] (4,-2) circle (2pt);\draw[fill=red,opacity=0.6] (6,-2) circle (2pt); \draw[fill=red,opacity=0.6] (8,-2) circle (2pt); \draw[fill=red,opacity=0.6] (6,-4) circle (2pt); \draw[fill=red,opacity=0.6] (8,-4) circle (2pt);\draw[fill=red,opacity=0.6] (8,-6) circle (2pt); \draw[fill=red,opacity=0.6] (4,2) circle (2pt);  \draw[fill=red,opacity=0.6] (2,6) circle (2pt); \draw[fill=red,opacity=0.6] (2,2) circle (2pt); \draw[fill=red,opacity=0.6] (2,4) circle (2pt); \draw[fill=red,opacity=0.6] (4,4) circle (2pt);\draw[fill=red,opacity=0.6] (6,2) circle (2pt);

\draw[fill=red,opacity=0.6] (-2,4) circle (2pt);\draw[fill=red,opacity=0.6] (-2,6) circle (2pt);\draw[fill=red,opacity=0.6] (-2,8) circle (2pt);\draw[fill=red,opacity=0.6] (-4,6) circle (2pt);\draw[fill=red,opacity=0.6] (-4,8) circle (2pt); \draw[fill=red,opacity=0.6] (-6,8) circle (2pt);

\draw[fill=red,opacity=0.6] (-4,2) circle (2pt); \draw[fill=red,opacity=0.6] (-6,4) circle (2pt); \draw[fill=red,opacity=0.6] (-6,2) circle (2pt);
\draw[fill=red,opacity=0.6] (-8,2) circle (2pt); \draw[fill=red,opacity=0.6] (-8,4) circle (2pt); \draw[fill=red,opacity=0.6] (-8,6) circle (2pt);
\draw[fill=red,opacity=0.6] (-2,-2) circle (2pt); \draw[fill=red,opacity=0.6] (-2,-4) circle (2pt); \draw[fill=red,opacity=0.6] (-2,-6) circle (2pt);
\draw[fill=red,opacity=0.6] (-4,-4) circle (2pt);\draw[fill=red,opacity=0.6] (-4,-2) circle (2pt); \draw[fill=red,opacity=0.6] (-6,-2) circle (2pt);

\draw[fill=red,opacity=0.6] (2,-4) circle (2pt);\draw[fill=red,opacity=0.6] (2,-6) circle (2pt); \draw[fill=red,opacity=0.6] (2,-8) circle (2pt);
\draw[fill=red,opacity=0.6] (4,-6) circle (2pt); \draw[fill=red,opacity=0.6] (4,-8) circle (2pt);\draw[fill=red,opacity=0.6] (6,-8) circle (2pt);
\end{tikzpicture}}
\caption{A well possed initial value problem (blue circles that lie on the initial lines $l=0,$ $m=0,$ and $n=0$) on the $Q(A2)$ lattice. Evaluation of the solution  (red circles) on the six chambers defined by the  initial lines $l=0,$ $m=0,$ and $n=0.$ The solution in neighbouring chambers is defined on triangles of opposite colour-the black triangles are represented by grey colour, while the white ones by yellow-} \label{figfin}
\end{figure} 

%\section{Conclusions}

\section*{Acknowledgements}
\parbox{.135\textwidth}{\begin{tikzpicture}[scale=.03]
\fill[fill={rgb,255:red,0;green,51;blue,153}] (-27,-18) rectangle (27,18);
\pgfmathsetmacro\inr{tan(36)/cos(18)}
\foreach \i in {0,1,...,11} {
\begin{scope}[shift={(30*\i:12)}]
\fill[fill={rgb,255:red,255;green,204;blue,0}] (90:2)
\foreach \x in {0,1,...,4} { -- (90+72*\x:2) -- (126+72*\x:\inr) };
\end{scope}}
\end{tikzpicture}} \parbox{.85\textwidth}
{This research is part of the project No. 2022/45/P/ST1/03998  co-funded by the National Science Centre and the European Union Framework Programme
 for Research and Innovation Horizon 2020 under the Marie Sklodowska-Curie grant agreement No. 945339. For the purpose of Open Access, the author has applied a CC-BY public copyright licence to any Author Accepted Manuscript (AAM) version arising from this submission.}

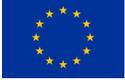
\appendix

\section{Proof of Proposition \ref{1st_prop}}\label{appa}

Mapping $Q$ is a $3D-$compatible map iff equations (\ref{3d:comp:def1}) are satisfied. In detail, applying the maps $Q_{ij}$ on the initial data $(x,y,z),$ we obtain
\begin{align*}
  (x,y,z)\xmapsto{Q_{12}}&(x_2,y_1,z),& (x,y,z)\xmapsto{Q_{13}}&(x_3,y,z_1),& (x,y,z)\xmapsto{Q_{23}}&(x,y_3,z_2),
\end{align*}
that is the six values $(x_2,x_3,y_1,y_3,z_1,z_2).$ Applying one more time the maps $Q_{ij}$ on these six values we obtain $(x_{23},x_{32},y_{13},y_{31},z_{12},z_{21}).$ The fact that $Q$ is a $3D-$compatible map assures $x_{23}=x_{32},$ $y_{13}=y_{31}$ and $z_{12}=z_{21}$, see Figure \ref{3dcmp}. The equivalence of items $(1)$ and $(3)$ is proven in \cite{ABS:YB}. 

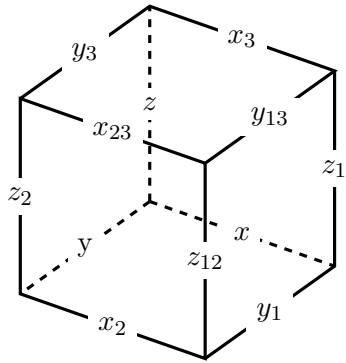
\begin{figure}[h]
  \centering

\tdplotsetmaincoords{60}{125}
\begin{tikzpicture}[
		tdplot_main_coords,
		grid/.style={very thin,gray},
		axis/.style={->,blue,thick},
		%cube/.style={very thick,fill=red},
cube/.style={very thick},
		cube hidden/.style={very thick,dashed}]
	
	%draw the front-right of the cube
	\draw[cube] (3,0,0) -- (3,3,0)node[midway,fill=white]{$x_2$} -- (3,3,3)node[midway,fill=white]{$z_{12}$} -- (3,0,3)node[midway,fill=white]{$x_{23}$} -- cycle node[midway,fill=white]{$z_2$};

%draw the top of the cube

\draw[cube]  (3,0,3) -- (0,0,3)node[midway,fill=white]{$y_3$}--(0,3,3)node[midway,fill=white]{$x_{3}$}-- (3,3,3)node[midway,fill=white]{$y_{13}$};

%draw the right of the cube
\draw[cube] (0,3,3)--(0,3,0)node[midway,fill=white]{$z_{1}$}--(3,3,0)node[midway,fill=white]{$y_{1}$};

	%draw the front-left of the cube
	%\draw[cube] (0,2,0) -- (2,2,0) -- (2,2,2) -- (0,2,2) -- cycle;

	%draw the top of the cube
	%\draw[cube] (0,0,2) -- (0,2,2) -- (2,2,2) -- (2,0,2) -- cycle;

	%draw dashed lines to represent hidden edges
	\draw[cube hidden] (0,0,0) -- (3,0,0)node[midway,fill=white]{y};
	\draw[cube hidden] (0,0,0) -- (0,3,0)node[midway,fill=white]{$x$};
	\draw[cube hidden] (0,0,0) -- (0,0,3)node[midway,fill=white]{$z$};
	
\end{tikzpicture}
  \caption{Representation on the cube of the   $3D-$compatible map $Q$}\label{3dcmp}
\end{figure}

In the following Figures we prove  that provided that $Q$ is a $3D-$compatible map, the map $R^{-1}:=(Q^c)^{-1},$ that is the inverse of the companion map $Q^c$ is a Yang-Baxter map (see item $(4)$).

\tdplotsetmaincoords{60}{125}
%\tdplotsetrotatedcoords{8}{8}{8} %<- rotate around (z,y,z)
\tdplotsetrotatedcoords{0}{20}{0} %<- rotate around (z,y,z)

\begin{figure}[htb]\adjustbox{scale=0.7}{
\begin{minipage}{.24\textwidth}
\begin{tikzpicture}[
		tdplot_main_coords,
		grid/.style={very thin,gray},
		axis/.style={->,blue,thick},
		%cube/.style={very thick,fill=red},
cube/.style={very thick},
		mynode/.style={circle,fill=green,minimum width=5pt,inner sep=3pt,outer sep=3pt}]
%\draw[cube hidden] (0,0,0) -- (3,0,0)node[midway,fill=white]{};
	\draw[dashed] (0,0,0) -- (0,3,0)node[midway,fill=white]{$x$};
%	\draw[cube hidden] (0,0,0) -- (0,0,3)node[midway,fill=white]{};
\draw[cube](0,3,0)--(3,3,0)node[midway,fill=white]{$y$};
\draw[cube](3,3,0) -- (3,3,3)node[midway,fill=white]{$z$};	
\end{tikzpicture}
\end{minipage}\hspace{-1.0cm}
$\xmapsto{R^{-1}_{12}}$\;
\begin{minipage}{.24\textwidth}
\begin{tikzpicture}[
		tdplot_main_coords,
		grid/.style={very thin,gray},
		axis/.style={->,blue,thick},
		%cube/.style={very thick,fill=red},
cube/.style={very thick},
		mynode/.style={circle,fill=green,minimum width=5pt,inner sep=3pt,outer sep=3pt}]
%\draw[cube hidden] (0,0,0) -- (3,0,0)node[midway,fill=white]{};
	\draw[dashed] (0,0,0) -- (0,3,0)node[midway,fill=white]{$x$};
%	\draw[cube hidden] (0,0,0) -- (0,0,3)node[midway,fill=white]{};
\draw[cube](0,3,0)--(3,3,0)node[midway,fill=white]{$y$};
\draw[cube](3,3,0) -- (3,3,3)node[midway,fill=white]{$z$};	
\draw[dashed] (0,0,0)--(3,0,0)node[midway,fill=white]{$y_{-1}$};\draw[cube] (3,0,0)--(3,3,0)node[midway,fill=white]{$x_{2}$};
\end{tikzpicture}
\end{minipage} \hspace{0.3cm} $\xmapsto{R^{-1}_{13}}$\;
\begin{minipage}{.24\textwidth}
\begin{tikzpicture}[
		tdplot_main_coords,
		grid/.style={very thin,gray},
		axis/.style={->,blue,thick},
		%cube/.style={very thick,fill=red},
cube/.style={very thick},
		mynode/.style={circle,fill=green,minimum width=5pt,inner sep=3pt,outer sep=3pt}]
%\draw[cube hidden] (0,0,0) -- (3,0,0)node[midway,fill=white]{};
	\draw[dashed] (0,0,0) -- (0,3,0)node[midway,fill=white]{$x$};
%	\draw[cube hidden] (0,0,0) -- (0,0,3)node[midway,fill=white]{};
\draw[cube](0,3,0)--(3,3,0)node[midway,fill=white]{$y$};
\draw[cube](3,3,0) -- (3,3,3)node[midway,fill=white]{$z$};	
\draw[dashed] (0,0,0)--(3,0,0)node[midway,fill=white]{$y_{-1}$};\draw[cube] (3,0,0)--(3,3,0)node[midway,fill=white]{$x_{2}$};
\draw[cube] (3,3,3)-- (3,0,3)node[midway,fill=white]{$x_{23}$}--(3,0,0)node[midway,fill=white]{$z_{-1}$};
\end{tikzpicture}
\end{minipage}
\hspace{0.6cm}  $\xmapsto{R^{-1}_{23}}$\;
\begin{minipage}{.24\textwidth}
\begin{tikzpicture}[
		tdplot_main_coords,
		grid/.style={very thin,gray},
		axis/.style={->,blue,thick},
		%cube/.style={very thick,fill=red},
cube/.style={very thick},
		mynode/.style={circle,fill=green,minimum width=5pt,inner sep=3pt,outer sep=3pt}]
%\draw[cube hidden] (0,0,0) -- (3,0,0)node[midway,fill=white]{};
	\draw[dashed] (0,0,0) -- (0,3,0)node[midway,fill=white]{$x$};
%	\draw[cube hidden] (0,0,0) -- (0,0,3)node[midway,fill=white]{};
\draw[cube](0,3,0)--(3,3,0)node[midway,fill=white]{$y$};
\draw[cube](3,3,0) -- (3,3,3)node[midway,fill=white]{$z$};	
\draw[dashed] (0,0,0)--(3,0,0)node[midway,fill=white]{$y_{-1}$};\draw[cube] (3,0,0)--(3,3,0)node[midway,fill=white]{$x_{2}$};
\draw[cube] (3,3,3)-- (3,0,3)node[midway,fill=white]{$x_{23}$}--(3,0,0)node[midway,fill=white]{$z_{-1}$};
\draw[cube]  (3,0,3) -- (0,0,3)node[midway,fill=white]{$y_{-13}$};
\draw[dashed] (0,0,0) -- (0,0,3)node[midway,fill=white]{$z_{-1-2}$};
\end{tikzpicture}
\end{minipage} }
\caption{The chain of maps $R^{-1}_{23}R^{-1}_{13}R^{-1}_{12}$ applied on the initial data $(x,y,z)$.}\label{fig2}
\end{figure}
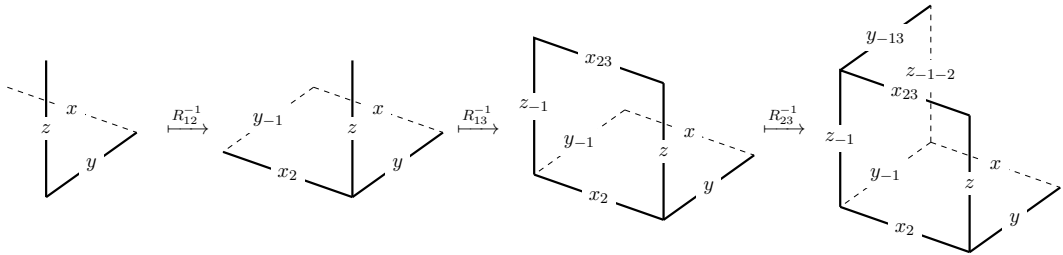

\begin{figure}[htb]\adjustbox{scale=0.7}{
\begin{minipage}{.24\textwidth}
\begin{tikzpicture}[
		tdplot_main_coords,
		grid/.style={very thin,gray},
		axis/.style={->,blue,thick},
		%cube/.style={very thick,fill=red},
cube/.style={very thick},
		mynode/.style={circle,fill=green,minimum width=5pt,inner sep=3pt,outer sep=3pt}]
%\draw[cube hidden] (0,0,0) -- (3,0,0)node[midway,fill=white]{};
	\draw[dashed] (0,0,0) -- (0,3,0)node[midway,fill=white]{$x$};
%	\draw[cube hidden] (0,0,0) -- (0,0,3)node[midway,fill=white]{};
\draw[cube](0,3,0)--(3,3,0)node[midway,fill=white]{$y$};
\draw[cube](3,3,0) -- (3,3,3)node[midway,fill=white]{$z$};	
\end{tikzpicture}
\end{minipage}\hspace{-1.0cm}
$\xmapsto{R^{-1}_{23}}$\hspace{0.2cm}
\begin{minipage}{.24\textwidth}
\begin{tikzpicture}[
		tdplot_main_coords,
		grid/.style={very thin,gray},
		axis/.style={->,blue,thick},
		%cube/.style={very thick,fill=red},
cube/.style={very thick},
		mynode/.style={circle,fill=green,minimum width=5pt,inner sep=3pt,outer sep=3pt}]
%\draw[cube hidden] (0,0,0) -- (3,0,0)node[midway,fill=white]{};
	\draw[dashed] (0,0,0) -- (0,3,0)node[midway,fill=white]{$x$};
%	\draw[cube hidden] (0,0,0) -- (0,0,3)node[midway,fill=white]{};
\draw[cube](0,3,0)--(3,3,0)node[midway,fill=white]{$y$};
\draw[cube](3,3,0) -- (3,3,3)node[midway,fill=white]{$z$};	
\draw[cube](0,3,0)--(0,3,3)node[midway,fill=white]{$z_{-2}$};
\draw[cube](0,3,3)--(3,3,3)node[midway,fill=white]{$y_{3}$};
\end{tikzpicture}
\end{minipage} \hspace{-0.3cm}
 $\xmapsto{R^{-1}_{13}}$\;
\begin{minipage}{.24\textwidth}
\begin{tikzpicture}[
		tdplot_main_coords,
		grid/.style={very thin,gray},
		axis/.style={->,blue,thick},
		%cube/.style={very thick,fill=red},
cube/.style={very thick},
		mynode/.style={circle,fill=green,minimum width=5pt,inner sep=3pt,outer sep=3pt}]
\draw[dashed] (0,0,0) -- (0,3,0)node[midway,fill=white]{$x$};
%	\draw[cube hidden] (0,0,0) -- (0,0,3)node[midway,fill=white]{};
\draw[cube](0,3,0)--(3,3,0)node[midway,fill=white]{$y$};
\draw[cube](3,3,0) -- (3,3,3)node[midway,fill=white]{$z$};	
\draw[cube](0,3,0)--(0,3,3)node[midway,fill=white]{$z_{-2}$};
\draw[cube](0,3,3)--(3,3,3)node[midway,fill=white]{$y_{3}$};
\draw[cube] (0,0,3) -- (0,3,3)node[midway,fill=white]{$x_3$};
\draw[dashed] (0,0,0) -- (0,0,3)node[midway,fill=white]{$z_{-1-2}$};
\end{tikzpicture}
\end{minipage}
\hspace{0.6cm}  $\xmapsto{R^{-1}_{12}}$\;
\begin{minipage}{.24\textwidth}
\begin{tikzpicture}[
		tdplot_main_coords,
		grid/.style={very thin,gray},
		axis/.style={->,blue,thick},
		%cube/.style={very thick,fill=red},
cube/.style={very thick},
		mynode/.style={circle,fill=green,minimum width=5pt,inner sep=3pt,outer sep=3pt}]
\draw[dashed] (0,0,0) -- (0,3,0)node[midway,fill=white]{$x$};
%	\draw[cube hidden] (0,0,0) -- (0,0,3)node[midway,fill=white]{};
\draw[cube](0,3,0)--(3,3,0)node[midway,fill=white]{$y$};
\draw[cube](3,3,0) -- (3,3,3)node[midway,fill=white]{$z$};	
\draw[cube](0,3,0)--(0,3,3)node[midway,fill=white]{$z_{-2}$};
\draw[cube](0,3,3)--(3,3,3)node[midway,fill=white]{$y_{3}$};
\draw[cube] (0,0,3) -- (0,3,3)node[midway,fill=white]{$x_3$};
\draw[dashed] (0,0,0) -- (0,0,3)node[midway,fill=white]{$z_{-1-2}$};
\draw[cube] (3,3,3)-- (3,0,3)node[midway,fill=white]{$x_{23}$};
\draw[cube]  (3,0,3) -- (0,0,3)node[midway,fill=white]{$y_{-13}$};
\end{tikzpicture}
\end{minipage} }
\caption{The chain of maps $R^{-1}_{12}R^{-1}_{13}R^{-1}_{23}$ applied on the initial data $(x,y,z)$.}\label{fig3}
\end{figure}
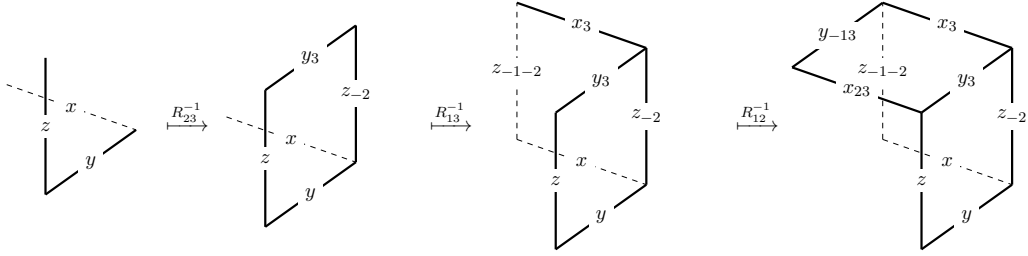
The assumption that $Q$ is $3D-$compatible assures that two ways (see Figure \ref{fig2} and \ref{fig3}) of obtaining $(x_{23},y_{-13},z_{-1-2})$ from the initial data $(x,y,z)$ lead to the same result and vice versa.

The proof of the rest of the items of this Proposition follows in a similar manner.

\section{Compatibility of the systems of equations  on a sector of fcc lattice}
Having three point relations it is more adequate to consider them on graphs that consist of triangles rather than of quadrilaterals. E.g. consider sector of fcc lattice
\[ {\mathcal S}=\left\{o+m_1 \vec{e_1} +m_2 \left(\frac{1}{2}\vec{e_1}+\frac{\sqrt{3}}{2} \vec{e_2} \right)+m_3
\left(\frac{1}{2}\vec{e}_1+\frac{\sqrt{3}}{6} \vec{e}_2 + \frac{\sqrt{6}}{3}  \vec{e}_3\right) \in {\mathbb E}^3 \, | \, m_1, m_2, m_3 \in {\mathbb N}\right\} \]
where $o$ is given point of Euclidean space (origin) and $( \vec{e}_1 , \vec{e}_2 , \vec{e}_3)$ is a given orthonormal basis of the tangent space.

Prescribing values of function $\phi$ at points $(m_1,0,0)$, $(0,m_2,0)$, $(0,0,m_3)$  where $m_1, m_2, m_3 \in {\mathbb N}$ and using  triangular recurrences we are talking about, one can propagate the solution to every point of 
${\mathcal S}$. Each of the systems $Q_{I}$,  $Q_{II}$ and $Q_{III}$ are of the form 
\[ \phi_{12}=f(\phi_1, \phi_2;{\mathbf a}^{12}) ,\quad \phi_{13} =g(\phi_1, \phi_3;{\mathbf a}^{13}), \quad \phi_{23}=h(\phi_2, \phi_3;{\mathbf a}^{23})\] 
where ${\mathbf a}^{ij}$ stands for set of parameters that depend only on $m_i$ and $m_j$ variables.
All the systems are compatible i.e. if we calculate
\[ \phi_{123}=f(\phi_{13}, \phi_{23};{\mathbf a}^{12}),\quad \phi_{132}=g(\phi_{12}, \phi_{23};{\mathbf a}^{13}), \quad \phi_{231}=h(\phi_{23}, \phi_{13};{\mathbf a}^{23}) \] 
then we will see that  \[\phi_{123} =\phi_{132}=\phi_{231}\]  holds. %We illustrate it on figure ???

%\bibliographystyle{unsrt}
%\bibliography{ref_0}

\end{document}